%% file: main.tex
\documentclass[letterpaper, 10 pt, conference]{ieeeconf}  % Comment this line out
\IEEEoverridecommandlockouts                              % This command is only
\usepackage{lilypack}

\title{\LARGE \bf
Energy Disaggregation via Adaptive Filtering
}

\author{Roy Dong, Lillian J. Ratliff, Henrik Ohlsson, and S. Shankar Sastry% <-this % stops a space
\thanks{R. Dong, L. Ratliff, H. Ohlsson,  and S. S. Sastry are with the
  Department of Electrical Engineering and Computer Sciences,
  University of California at Berkeley, CA, USA
  {\tt\small\{roydong,ratliffl,ohlsson,sastry\}}
  {\tt\small@eecs.berkeley.edu}}
\thanks{H. Ohlsson is also with the Division of Automatic Control, Department of Electrical Engineering, Link\"oping University, Sweden.}%
\thanks{The work presented is supported by the NSF
Graduate Research Fellowship under grant DGE 1106400, NSF
CPS:Large:ActionWebs award number 0931843, TRUST (Team for Research in
Ubiquitous Secure Technology) which receives support from NSF (award
number CCF-0424422), and FORCES (Foundations Of Resilient
CybEr-physical Systems), the European Research Council
   under the advanced grant LEARN, contract 267381, a postdoctoral grant from the Sweden-America
   Foundation, donated by ASEA's Fellowship Fund, and  by a postdoctoral
   grant from the Swedish Research Council.}%
}

\begin{document}

\maketitle
\thispagestyle{empty}
\pagestyle{empty}

%%%%%%%%%%%%%%%%%%%%%%%%%%%%%%%%%%%%%%%%%%%%%%%%%%%%%%%%%%%%%%%%%%%%%%%%%%%%%%%%
\begin{abstract}
  %TODO:
%It has been shown that the energy consumption of residential
%and commercial buildings  can be reduced with as
%much as 20\% by simply changing people's behavior and through better
%feedback. The costly installation of plug meters has been holding this
%alternative back and encouraged   cheaper alternatives. One very
%promising technique is energy disaggregation. The goal of energy
%disaggregation is to replace physical plug load sensors with cheap software
%sensors. 
The energy disaggregation problem is recovering device level power consumption signals from the aggregate power consumption signal for a building. 
We show in this paper how the disaggregation problem
can be reformulated as an adaptive filtering
problem. This gives both a novel disaggregation algorithm and 
a better theoretical understanding for disaggregation. In particular,
we show how the disaggregation problem can be solved online using a 
filter bank and discuss its optimality.
\end{abstract}

%%%%%%%%%%%%%%%%%%%%%%%%%%%%%%%%%%%%%%%%%%%%%%%%%%%%%%%%%%%%%%%%%%%%%%%%%%%%%%%%
\section{INTRODUCTION}
\label{sec:introduction}
\input{intro}

%TODO: MAYBE RE-WRITE THIS ONCE THE OTHER STUFF IS DONE.

%%%%%%%%%%%%%%%%%%%%%%%%%%%%%%%%%%%%%%%%%%%%%%%%%%%%%%%%%%%%%%%%%%%%%%%%%%%%%%%%
\section{BACKGROUND}
\label{sec:background}
\input{back}

%%%%%%%%%%%%%%%%%%%%%%%%%%%%%%%%%%%%%%%%%%%%%%%%%%%%%%%%%%%%%%%%%%%%%%%%%%%%%%%%
\section{PROBLEM FORMULATION}
\label{sec:problem_formulation}
\input{framework}

%%%%%%%%%%%%%%%%%%%%%%%%%%%%%%%%%%%%%%%%%%%%%%%%%%%%%%%%%%%%%%%%%%%%%%%%%%%%%%%%
\section{PROPOSED FRAMEWORK}
\label{sec:big_framework}
\input{disagg_adapt_filt}

%%%%%%%%%%%%%%%%%%%%%%%%%%%%%%%%%%%%%%%%%%%%%%%%%%%%%%%%%%%%%%%%%%%%%%%%%%%%%%%%

\section{ENERGY DISAGGREGATION VIA ADAPTIVE FILTERING}
\label{sec:implementation}
\input{implementation}
%%%%%%%%%%%%%%%%%%%%%%%%%%%%%%%%%%%%%%%%%%%%%%%%%%%%%%%%%%%%%%%%%%%%%%%%%%%%%%%%

\section{THEORY}
\label{sec:theory}
\input{theoretical}

%%%%%%%%%%%%%%%%%%%%%%%%%%%%%%%%%%%%%%%%%%%%%%%%%%%%%%%%%%%%%%%%%%%%%%%%%%%%%%%%
\section{EXPERIMENT}
\label{sec:experiment}
\input{experiment}

%%%%%%%%%%%%%%%%%%%%%%%%%%%%%%%%%%%%%%%%%%%%%%%%%%%%%%%%%%%%%%%%%%%%%%%%%%%%%%%%
\section{CONCLUSIONS AND FUTURE WORK}
\label{sec:conclusions}
\input{conc}

%%%%%%%%%%%%%%%%%%%%%%%%%%%%%%%%%%%%%%%%%%%%%%%%%%%%%%%%%%%%%%%%%%%%%%%%%%%%%%%%
\section*{ACKNOWLEDGMENTS}
The authors would like to thank Professor Lennart Ljung for many insights into this problem and Aaron Bestick for countless discussions and assistance in experimental setup.

%%%%%%%%%%%%%%%%%%%%%%%%%%%%%%%%%%%%%%%%%%%%%%%%%%%%%%%%%%%%%%%%%%%%%%%%%%%%%%%%

\bibliographystyle{IEEEtran}
\bibliography{main}

\end{document}

%% file: intro.tex
Power consumption data of individual devices have the potential to greatly decrease costs in the electricity grid. Currently, residential and commercial buildings account for 40\% of total energy consumption~\cite{Perez-Lombard2008}, and studies have estimated that 20\% of this consumption could be avoided with efficiency improvements with little to no cost~\cite{Creyts2007,Laitner2009}. It is believed that the largest barrier to achieving these energy cost reductions is due to behavioral reasons~\cite{Crabtree2008}. 

The authors of~\cite{Armel2013} claim that the full potential of the new smart meter technology cannot be exploited if human factors are not considered; that is, we must recognize that the smart grid is a system with a human in the loop. Furthermore, the authors note that billions of dollars are being expended on the installation of smart meters, which can provide the utility company with high resolution data on a building's power consumption. However, this hardware currently only provides the aggregate power consumption data, and deployment is at a sufficiently advanced stage that a change in hardware is prohibitively expensive.

%and deployment is at an sufficiently advanced stage that a change in hard
%The authors of~\cite{Armel2013} point out that billions of dollars are being expended on the installation of smart meters. The authors further claim that the full potential of this new technology cannot be exploited if human factors are not considered; that is, we must recognize that the smart grid is a system with a human in the loop. 
Disaggregation presents a way in which consumption patterns of individuals can be learned by the utility company. This information would allow the utility to present this information to the consumer, with the goal of increasing consumer awareness about energy usage. Studies have shown that this is sufficient to improve consumption patterns~\cite{Ehrhardt-Martinez2010}. %Further, it can be used to develop incentives for behavior modification resulting in decreased consumption and demand shifting.

Outside of informing consumers about ways to improve energy efficiency, disaggregation presents an opportunity for utility companies to strategically market products to consumers. It is now common practice for companies to monitor our online activity and then present advertisements which are targeted to our interests. This is known as `personalized advertising'. Disaggregation of energy data provides a means to similarly market products to consumers. This leads to the question of user privacy and the question of ownership with regards to power consumption information.
%the realization that there are inherent trade-offs between fine grained disaggregation and consumer privacy. 
Treatment of the issue of consumer privacy in the smart grid is outside the scope of this paper. However, this is discussed in~\cite{cardenas2012}.

Additionally, disaggregation also presents opportunities for improved control. Many devices, such as heating, ventilation, and air conditioning (HVAC) units in residential and commercial buildings implement control policies that are dependent on real-time measurements. Disaggregation can provide information to controllers about system faults, such as device malfunction, which may result in inefficient control. It can also provide information about energy usage which is informative for demand response programs.

Our aim in this paper is to formulate the disaggregation problem in the filter banks framework. In doing so we extend our previous work in which we developed a method that combines the use of generative models, e.g. linear dynamical models of devices, with a supervised approach to disaggregation~\cite{Dong2013}.  In particular, we develop an algorithm for disaggregation of whole building energy data using dynamical models for devices and filter banks for determining the most likely inputs to the dynamical models. Under mild assumptions on the noise characteristics we are able to provide guarantees for when the algorithm recovers the disaggregated signal that most closely matches our observed data and priors.%an optimal solution. 

In Section \ref{sec:background}, we discuss previous work on the topic of energy disaggregation. In Section \ref{sec:problem_formulation}, we formally define the problem of energy disaggregation. In Sections \ref{sec:dev_mod} to \ref{sec:disagg}, we establish our framework for solving the problem of energy disaggregation. In Section \ref{sec:implementation}, we provide an online adaptive filtering algorithm for estimating individual device power consumption patterns, and in Section \ref{sec:theory}, we prove properties of this algorithm. 
% problemtheoretical results for our framework. In Section \ref{sec:implementation}, we discuss how to implement our framework as an algorithm which can be run online.
In Section \ref{sec:experiment}, we show energy disaggregation results from a small-scale experiment. Finally, in Section \ref{sec:conclusions}, we give concluding remarks and describe plans for future work.

%TODO: emphasize some of the human in the loop benefits that were lauded in the Holy Grail Stanford paper.

%% file: back.tex
%Rewrite the background as before, but with more of an emphasis on how the past work divides into those that focus on the transients and those that focus on being unsupervised like HMMs. Not much change is needed.
%
%TODO
%
%Also, add reference to CDC submission, paragraph explaining the novel contribution of this paper.
%The problem of energy disaggregation is not a new one. 
The problem of energy disaggregation, and the existing hardware for disaggregation, has been studied extensively in the literature (see~\cite{berges2010:lj, berges2009:gd}, for example).
%Another name for disaggregation is non-intrusive load monitoring (NILM); the problem of NILM and the existing hardware for NILM has been studied extensively in the literature (see~\cite{berges2010:lj, berges2009:gd}, for example). %The general consensus is that NILM is a method to present the consumer with information that makes them aware of their usage and potentially provides them insight into how to improve the efficiency of their usage. Further, the technology to perform NILM is becoming widely available. Hence, there is a need for flexible and efficient disaggregation algorithms. 
%As mentioned in Section \ref{sec:introduction}, the technology to measure the data required for energy disaggregation is becoming widely available. 
%NILM aims to present the consumer with information that makes them aware of their usage and potentially provides them insight into how to improve the efficiency of their usage. Further, the technology to perform NILM is becoming widely available. 
%Hence, there is a need for flexible and efficient disaggregation algorithms. 
%Disaggregation of energy data has emerged as one possible solution for identifying consumer behavior patterns and device malfunctions which lead to inefficient usage of energy.
The goal of the current disaggregation literature is to present methods for improving energy monitoring at the consumer level without having to place sensors at device level, but rather use existing sensors at the whole building level. 
%The motivation for such research is the need for improvement of energy efficiency in residential and commercial buildings which are major consumers. 
The concept of disaggregation is not new; however, only recently has it gained attention in the energy research domain, likely due to the emergence of smart meters and big data analytics, as discussed in Section \ref{sec:introduction}.

Disaggregation, in essence, is a single-channel source separation problem. The problem of recovering the components of an aggregate signal is an inverse problem and as such is, in general, ill-posed. Most disaggregation algorithms are batch algorithms and produce an estimate of the disaggregated signals given a batch of aggregate recordings.  There have been a number of survey papers summarizing the existing methods (e.g. see \cite{Zeifman2011:fh}, \cite{Kolter2011}). In an effort to be as self-contained as possible, we try to provide a broad overview of the existing methods and then explain how the disaggregation method presented in this paper differs from existing solutions.  

The literature can be divided into two main approaches, namely, supervised and unsupervised. Supervised disaggregation methods require a disaggregated data set for training. This data set could be obtained by, for example, monitoring typical appliances using plug sensors. Supervised methods assume that the variations between signatures for the same type of appliances is less than that between signatures of different types of appliances. Hence, the disaggregated data set does not need to be from the building that the supervised algorithm is designed for. However, the disaggregated data set must be collected prior to deployment, and come from appliances of a similar type to those in the target building. Supervised methods are typically discriminative. 

Unsupervised methods, on the other hand, do not require a disaggregated data set to be collected. They do, however, require hand tuning of parameters, which can make it hard for the methods to be generalized in practice. It should be said that also supervised methods have tuning parameters, but these can often be tuned using the training data. 

The existing supervised methods include sparse coding \cite{Kolter2010}, change detection and clustering based approaches \cite{drenker1999:rk,rahayu2012:sk} and pattern recognition \cite{farinaccio1999:lk}. The sparse coding approach tries to reconstruct the aggregate signal by selecting as few signatures as possible from a library of typical signatures. Similarly, in our proposed framework we construct a library of dynamical models and reconstruct the aggregate signal by using as few as possible of these models.  

The existing unsupervised methods include factorial hidden Markov models (HMMs), difference hidden Markov models and variants \cite{kim11,Kolter2012,Johnson2012,parson2012:AAAI:lk,pattem2012:gk} and temporal motif mining \cite{Shao2012}. Most unsupervised  methods model the on/off sequences of appliances using some variation of HMMs. These methods do not directly make use of the signature of a device and assume that the power consumption is piecewise constant. 

All methods we are aware of lack the use of %knowledge about the physical system, meaning 
the dynamics of the devices. While the existing supervised methods often do use device signatures, these methods are discriminative and an ideal method would be able to generate a power consumption signal from a given consumer usage profile. 
%have a dynamical model that is capable to generating a device signature given a combination of initial state and input. (TODO: This might be a little to specific to our framework to be a statement we can say about other approaches?) 
Both HMMs and linear dynamical models are generative as opposed to discriminative, making them more advantageous for modeling complex system behavior. In the unsupervised domain, HMMs are used; however, they are not estimated using data and they do not model the signature of a device. 

In a previous paper we developed a method which combines the use of generative models, i.e. linear dynamical models of devices, with a supervised approach to disaggregation~\cite{Dong2013}. In this paper, we extend previous work by formalizing our method within an adaptive filtering framework. Specifically, we formulate hypotheses on the on/off state of the devices over the time horizon for which we have data. The on/off state corresponds to whether the input is activated or not. Using filter banks and the dynamical models we have for device behavior, we evaluate which is the most likely hypothesis on the inputs. We provide an algorithm for this process. Under mild assumptions on the noise characteristics we are able to provide guarantees for when the our algorithm results in an optimal solution. The filter bank framework is similar to HMM frameworks in the sense that both methods essentially formulate hypotheses on which devices are on at each time instant. However, in contrast to HMMs, in the filter bank framework we incorporate the use of dynamical models to capture the transients of the devices, which helps identify them.

%% file: framework.tex
%In \cite{Dong2013}, we established a framework for energy disaggregation which draws on results in systems theory. In this paper, we present theoretical results which follow from our problem formulation.

%\subsection{Problem statement}

In this section, we formalize the problem of energy disaggregation.

Suppose we are given an aggregated power consumption signal for a building. We denote this data as $y[t]$ for $t = 0,1,\dots,T$, where $y[t]$ is the aggregate power consumption at time $t$. The entire signal will be referred to as $y$. This signal is the aggregate of the power consumption signal of several individual devices:
\begin{equation}
	y[t] = \sum_{i = 1}^D y_i[t] \text{ for } t = 0,1,\dots,T,
\end{equation}
where $D$ is the number of devices in the building and $y_i[t]$ is the power consumption of device $i$ at time $t$. The goal of disaggregation is to recover $y_i$ for $i = 1,2,\dots,D$ from $y$.

%Stated in this fashion, the problem is ill-posed. 
To solve this problem, it is necessary to impose additional assumptions on the signals $y_i$ and the number of devices $D$.

%%%%%%%%%%%%%%%%%%%%%%%%%%%%%%%%%%%%%%%%%%%%%%%%%%%%%%%%%%%%%%%%%%%%%%%%%%%%%%%%%%%%%%%%%%%%%%%%%%%%%%%%%%%%%%%%%%%%%%%%%%%%%%%

%% file: disagg_adapt_filt.tex
At a high level, our framework can be summarized as follows. First, in
the training phase of our disaggregation framework, we assume we have access to a training set of individual device power consumption data that is representative of the devices in the buildings of concern. From this training data, we build a library of models for individual devices. With these models, the disaggregation step becomes finding the most likely inputs to these devices that produces our observed output, the aggregate power consumption signal.

\subsection{Training phase}
\label{sec:dev_mod}

Suppose we have a training data set, which consists of the power consumption signals of individual devices. Let $z_i[t]$ for $t = 0,1,\dots,T_i$ be a power consumption signal for a device $i$. Then, $\{z_i\}_{i = 1}^D$ is our training data. %for devices $i = 1,2,\dots,D$ denote this training set. 
%Furthermore, we currently assume that every device in the household is represented in this training data set. This is an assumption we intend to relax in future work.
From this training data, we will learn models for individual devices.

For device $i$, we assume the dynamics take the form of a finite impulse response (FIR) model:
\begin{equation}
\label{eq:fir}
	z_i[t] = \sum_{j = 0}^{n_i} b_{i,j} u_i^z[t-j] + e_i[t],
\end{equation}
where $n_i$ is the order of the FIR model corresponding to device $i$, $b_{i,j}$ represent the parameters of the FIR model and $e_i[t]$ is white noise, i.e.\ random variables that are zero mean, finite variance, and independent across both time and devices. Furthermore, $u_i^z[t]$ represents the input to device $i$ at time $t$ in the training dataset, $z$.
% autoregressive model with exogenous inputs (ARX):
%\begin{equation}
%\label{eq:arx1}
%	z_i[t] = \sum_{j = 1}^{n_i}a_{i,j} z_i[t-j] + b_i u_i[t] + e_i[t],
%\end{equation}
%% \begin{equation}
%% 	y_i[t] = \theta_i^\top \psi_i[t] + e_i[t],
%% \end{equation}
%%where $\psi_i[t] = $ and $\theta_i$ is ...
%%These are known as autoregressive models with exogenous inputs (ARX models).
%where $a_{i,j}$ for $j = 1,2,\dots,n_i$ and $b_i$ represent the parameters of the ARX model, and $e_i[t]$ is white noise, i.e.\ zero mean, finite variance, and independent across both time and devices. 
%Furthermore, $u_i[t]$ represents the input at the $t$. %We also assume that the errors $e_i$ are zero-mean and independent across both time and devices.
%For simplicity, we will assume that the initial conditions are zero; that is, $z_i[t] = 0$ for $t = -1,-2,\dots,-n_i$.

We now make the following assumption:
\begin{assumption}
\label{ass:1}
FIR models fed by piecewise constant inputs give a rich enough setup to model the
energy consumed by individual appliances.
\end{assumption}

Firstly, many electrical appliances can be seen having a piecewise
constant input. %In our framework, the input represents the device
                %setting, or the command given to the device. 
For example, the input of  a conventional oven can be seen as
$0^\circ$F if the oven is off, and  $300^\circ$F if the oven is set to
heat to $300^\circ$F. 
Note that the input is not the actual internal temperature of the
oven, but rather the temperature setting on the oven. Since the
temperature setting is relatively infrequently changed, the input is
piecewise constant over time. Many other appliances are either on or
off, for example lights, and can be seen having a binary input with
infrequent changes. This is also a piecewise constant input. For a
washing machine, we have a discrete change between modes (washing,
spinning, etc.) and this mode sequence can be seen as the
piecewise constant input of the washing machine. 

Secondly, a FIR model can fit arbitrarily complex stable linear
dynamics. Assuming that FIR models fed by piecewise constant inputs
give a rich enough setup to model the energy consumed by individual
appliances is therefore often sufficient for energy disaggregation.

Thirdly, without any assumption on the inputs, the disaggregation problem later presented in Section~\ref{sec:disagg} is ill-posed; thus, Assumption~\ref{ass:1}, which assumes that changes in input are sparse, serves as a regularization which helps make the problem less ill-posed.

%With the input modeled in this fashion, we assume that the inputs are
%piecewise %constant over time. 
In most applications, we will not have access to any input data. Thus, our system identification step becomes estimation of both the input and the FIR parameters. This is known as a blind system identification problem, and is generally very difficult.

However, with the assumption that the inputs represent an on/off sequence, we can use simple change detection methods to estimate the binary input  %can use simple change detection methods to estimate an input signal
$u_i^z$. For more complicated inputs, we refer to \cite{Ohlsson2013}.

% This can be done if we have access to device-level measurements. In \cite{Ohlsson2013}, we develop a more sophisticated algorithm which assumes that the changes in input are sparse and relaxes the blind identification problem into a convex formulation which provably recovers the correct solution under certain conditions.

Although $n_i$ is not known a priori, we can select the value of $n_i$ using criterion from the system identification and statistics literature. For example, one can use the Akaike information criterion (AIC) or the Bayesian information criterion (BIC). For more information on model selection, as well as other possible criteria for model selection, we refer the reader to \cite{Ljung1999}.

%If we assume that the input is piecewise constant, with segments longer than $n_b + 1$, then, without loss of generality, we can rewrite the ARX model as the following:
Finally, we can succinctly rewrite \eqref{eq:fir} in vector form:
\begin{equation}
\label{eq:fir_vec}
	z_i[t] = \betab_i^\top \xib_i[t] + e_i[t],
\end{equation}
where $\betab_i$ are the FIR parameters:
\begin{equation}
\label{eq:beta_def}
	\betab_i =
	\begin{bmatrix}
		b_{i,0} 	& b_{i,1} 		& \dots	 	& b_{i,n_i}
	\end{bmatrix}^\top,
\end{equation}
and $\xib_i[t]$ are the regressors at time $t$:
\begin{equation}
	\xib_i[t] = 
	\begin{bmatrix}
		u_i^z[t] 	& u_i^z[t-1] 	& \dots 	& u_i^z[t-n_i]
	\end{bmatrix}^\top.
\end{equation}
%\begin{equation}
%\label{eq:arx_vec}
%	z_i[t] = \alphab_i^\top \zetab_i[t] + b_i u_i[t] + e_i[t],
%\end{equation}
%where $\alphab_i$ are the parameters:
%\begin{equation}
%\label{eq:alpha_def}
%	\alphab_i[t] =
%	\begin{bmatrix}
%		a_{i,1} 	& a_{i,2} 		& \dots	 	& a_{i,n_i}
%	\end{bmatrix}^\top,
%\end{equation}
%and $\zetab_i[t]$ is a vector of the past states at time $t$:
%\begin{equation}
%\label{eq:psi_defz}
%	\zetab_i[t] =
%	\begin{bmatrix}
%		z_i[t-1] 	& z_i[t-2] 	& \dots 	& z_i[t-n_i]
%	\end{bmatrix}^\top.
%\end{equation}

%%%%%%%%%%%%%%%%%%%%%%%%%%%%%%%%%%%%%%%%%%%%%%%%%%%%%%%%%%%%%%%%%%%%%%%%%%%%%%%%%%%%%%%%%%%%%%%%%%%%%%%%%%%%%%%%%%%%%%%%%%%%%%%

\subsection{Energy disaggregation}
\label{sec:disagg}

Suppose we have estimated a library of models for devices $i = 1,2,\dots,D$. That is, we are given $\betab_i$ for devices $i = 1,2,\dots,D$. 
Furthermore, we are given $y$. We wish to find $y_i$ for $i = 1,2,\dots,D$. Now, we make the following assumption:
%. That is, we are given $\alphab_i$ and $b_i$ for $i = 1,2,\dots,D$. Furthermore, we are given $y$. %For simplicity, we will assume that the initial conditions are zero; that is, $y_i[t] = 0$ for $t = -n_i,-n_i+1,\dots,-1$ and $i = 1,2,\dots,D$.
\begin{assumption}
\label{ass:all_modeled}
	The devices in our building are a subset of the devices $\{1,2,\dots,D\}$. Furthermore, these devices have dynamics of the form in \eqref{eq:fir_vec}.
\end{assumption}
Note here that we assume that all devices are modeled in our library,
or, equivalently, all devices are represented in our training
data. This is a common assumption in the disaggregation literature but
we plan to relax this assumption in future work. %Second, we assume

Now, this problem is equivalent to finding inputs to our devices that generate our observed aggregated signal. %$u_i[t]$ for $t = -n_i,n_i+1,,\dots,T$ and $i = 1,2,\dots,D$. 
More explicitly, let:
\begin{equation}
	\betab =
	\begin{bmatrix}
		\betab_1^\top & \betab_2^\top & \dots & \betab_D^\top
	\end{bmatrix}^\top,
\end{equation}
\begin{equation}
	\psib[t] =
	\begin{bmatrix}
		\psib_1[t]^\top & \psib_2[t]^\top & \dots & \psib_D[t]^\top
	\end{bmatrix}^\top,
\end{equation}
where $\betab_i$ are as defined in \eqref{eq:beta_def} for each device $i = 1,2,\dots,D$ and:
\begin{equation}
\label{eq:psi_defy}
	\psib_i[t] =
	\begin{bmatrix}
		u_i[t] 	& u_i[t-1] 	& \dots 	& u_i[t-n_i]
	\end{bmatrix}^\top.
\end{equation}
Then, 
%where $\alphab_i$, $\psib_i[t]$ are as defined in Equations \eqref{eq:alpha_def} and \eqref{eq:psi_def}.
we have a model for the aggregate power signal:
\begin{equation}
\label{eq:agg_model1}
	y[t] = \betab^\top \psib[t] + e[t],
\end{equation}
where $e[t] = \sum_{i = 1}^D e_i[t]$ is still white noise. For simplicity, we assume zero initial conditions, i.e.\ $u_i[t] = 0$ for $t = -n_i, -n_i+1, \dots, -1$. This assumption can easily be relaxed. Thus, the problem of energy disaggregation is now finding $u_i[t]$ for $t = 0,1,\dots,T$ and $i = 1,2,\dots,D$. 
%the problem of energy disaggregation is now finding $u_i[t]$ for $t = -n_i,n_i+1,,\dots,T$ and $i = 1,2,\dots,D$. 
Let:
\begin{equation}
	\ub[t] = 
	\begin{bmatrix}
		u_1[t] 	& u_2[t] & \dots & u_D[t]
	\end{bmatrix}^\top.
\end{equation}
Recall that in the training phase we assumed that FIR models fed by
piecewise constant inputs gave a rich enough setup for accurately modeling energy
consumption of individual appliances. We will in the disaggregation
step similarly assume that $u_i[t], i=1,\dots,D,$ are piecewise
constant over time. It follows that the vector-valued function $\ub$ is piecewise constant. 

Let a segment be defined as an interval in which $\ub$ is constant. Then, energy disaggregation becomes a segmentation problem. More formally, define a segmentation as $k^n = (k_1,k_2,\dots,k_n)$ such that $0 \leq k_1 < k_2 < \dots < k_n$. Here, both $n$ and $k_l,\,l=1,\dots,n,$ are unknown. For a segmentation $k^n$, we have that:
\begin{equation}
	\ub[s] = \ub[t] \text{ for all } k_{l-1} < s,t \leq k_l,
\end{equation}
with $k_0 = -1$.

Here, we will introduce some additional notation which will be helpful for the rest of this paper.

First, we introduce an alternative notation for segmentations. Let $\delta[t] = 1$ if $\ub[t] \neq \ub[t-1]$, and $0$ otherwise. In other words, $\delta[t]$ is a binary variable that equals $1$ if and only if the input changes between times $t-1$ and $t$. Thus, $k^n = (k_1,k_2,\dots,k_n)$ and $\delta \in \{0,1\}^T$ are equivalent representations of a segmentation. %Note that the $k^n$ and $\delta$ notations are equivalent; 
Throughout this paper we shall freely move between the two.

Next, suppose we are given a segmentation $k^n$. Then for each device $i$, we can define a function $\bar u_i : \{1,2,\dots,n\} \rightarrow \R$ such that $\bar u_i(l) = u_i[k_l]$, i.e.\ $\bar u_i(l)$ represents the input to device $i$ in the $l$th segment. Then, let $\ubarb : \{1,2,\dots,n\} \rightarrow \R^D, l \mapsto (\bar u_1(l), \bar u_2(l), \dots, \bar u_D(l))$. $\ubarb(l)$ represents the input to all devices in the $l$th segment.

Also, let $y^t$ denote all measurements available at time $t$. That is, $y^t = (y[0],y[1],\dots,y[t])$.

Let $p(\ub)$ denote a probability distribution on the user's input to the devices; that is, $p(\ub)$ is the likelihood of the input $\ub$. 
This encapsulates our prior on user consumption patterns. For example, in residential buildings, power consumption tends to be low mid-day, while in commercial buildings, power consumption drops off after work hours. This knowledge can be represented in $p(\ub)$. %In other words, $p(\ub)$ contains information on how often our piecewise constant input shifts, and information on what values it shifts to.

The disaggregation problem is to find the maximum a posteriori (MAP)
estimate of $\ub$ and, consequently, the power consumption of device $i$, %the power consumption of devise $i$ via $\betab_i^\top \psib_i[t]$)
given our observations. In Section \ref{sec:implementation}, we provide an adaptive filtering algorithm for solving this problem, and in Section \ref{sec:theory}, we provide theoretical guarantees of our proposed algorithm.

%Let $p(k^n)$ denote a probability distribution on $k^n$ that represents the prior we have on the segmentation. Note that for the energy disaggregation problem, this prior encapsulates our knowledge of user consumption patterns. If we do not have any prior knowledge about usage patterns, we can simply take $p(k^n \vert n)$ to be a uniform distribution and $p(n)$ to be an increasing function that penalizes large $n$. (In the absence of this sort of regularization, we would always choose the trivial answer where $n = T$.)
%
%Our goal is to find the maximum a posteriori (MAP) estimate of $k^n$ given our observations. We discuss theory for this problem in Section \ref{sec:theory}, and techniques for implementing the MAP estimator in Section \ref{sec:implementation}.

% We also introduce the following notation:
% \begin{equation}
% 	\ub[t+1] = (1 - \delta[t])\ub[t] + \delta[t] \vb[t]
% \end{equation}
% where $\delta[t] \in \{0,1\}$ represents whether or not the input changes between times $t$ and $t+1$, and $\vb[t] \in \R^D$ is an unknown vector that represents the new input at time $t+1$ if $\delta[t] = 1$.

%There are two broad categories of criteria by which we can select a segmentation. The first category is statistical criterion, which often finds the maximum likelihood or maximum a posteriori (MAP) estimate of $k^n$. The second category is information based criterion, which includes AIC and BIC. This problem has a mature 

There are many  criteria other than the MAP by which to select a segmentation. The best criteria for selection of a segmentation is an active topic of debate, and a thorough treatment of this question is outside the scope of this paper. We refer the interested reader to \cite{Gustafsson2000} for more details on  segmentation.

%% file: implementation.tex
\subsection{Algorithm definition}

In this section, we provide a tractable algorithm to solve the problem posed in Section \ref{sec:big_framework}. Furthermore, this algorithm is defined recursively on measurements across time, so it can be run online.

We draw on results in the adaptive filtering literature. An adaptive filter is any filter that adjusts its own parameters based on observations. In our particular case, we use a filter bank approach to handle the problem presented in Section~\ref{sec:big_framework}. A filter bank is a collection of filters, and the adaptive element of a filter bank is in the insertion and deletion of filters, as well as the selection of the optimal filter.

We will define a filter bank, and also the problem a filter bank solves. Suppose we are given measurements $y^t$. We wish to find the maximum a posteriori estimate of the input $\ub$ given our measurements $y^t$: we wish to find $\ub$ that maximizes $p(\ub \vert y^t)$, which is equivalent to maximizing $p( y^t \vert \ub) p(\ub)$. Decomposing $u$ into $\delta$ and $\ubarb$, we can again rewrite this as maximizing $p(y^t \vert \ubarb,\delta) p(\ubarb \vert \delta)p(\delta)$. Note that we can calculate:
\begin{equation}
p(\delta) = \int p(\ubarb, \delta) d\ubarb.
\end{equation}
The final manipulation is that we wish to find a $\delta$ to maximize the following quantity:
\begin{equation}
\label{eq:ubarb_max}
	\max_{\ubarb}\ p(y^t \vert \ubarb,\delta) p(\ubarb \vert \delta)p(\delta).
\end{equation}

Now, our algorithm maintains a collection of filters, known as a filter bank. Let $\Fc$ denote this filter bank. Each filter $f \in \Fc$ corresponds to a segmentation $\delta_f \in \{0,1\}^t$. Given each $\delta_f$, we can calculate:
\begin{equation}
	\ubarb_f = \argmax_{\ubarb}\ p(y^t \vert \ubarb,\delta_f) p(\ubarb \vert \delta_f)p(\delta_f).
\end{equation}
\begin{equation}\label{eq:posterior}
	p_f = \max_{\ubarb}\ p(y^t \vert \ubarb,\delta_f) p(\ubarb \vert \delta_f)p(\delta_f).
\end{equation}
There are only finitely many possible $\delta$. Thus, if we kept a
filter $f$ for every possible segmentation $\delta$, we could easily
find the MAP estimate of $\ub$. However, the filter bank $\Fc$ would grow exponentially with time. Thus, we need to find methods to keep the size of $\Fc$ under control.

%These techniques are best understood from the perspective of an online algorithm.
The process of finding the best segmentation can be seen as exploring a binary tree. That is, a segmentation $\delta$ can be thought of as a leaf node on a binary tree of depth $t$. This is visualized in Figure \ref{fig:binary_tree}.

\begin{figure}[ht]
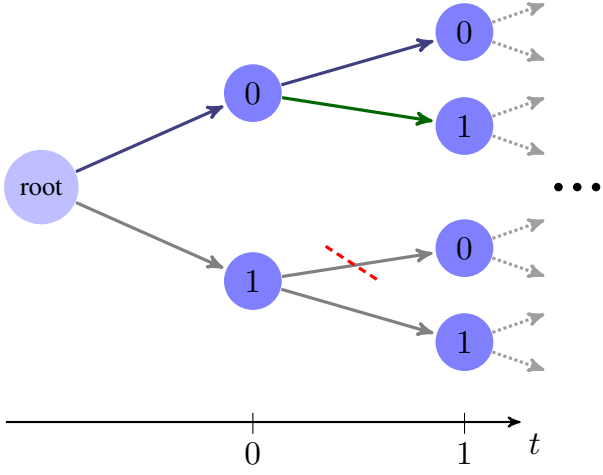

%	\begin{center}
          \include{fig/binaryTree}
          %\includegraphics[width=\columnwidth]{fig/binary_tree4.pdf}
%	\end{center}
	\caption{A segmentation $\delta$ can be thought of as a leaf node on a binary tree of depth $t$. That is, $\delta$ corresponds exactly to one leaf node of this binary tree. %Consider the path $(0)$. If we choose not to branch $(0)$, then only the blue path $(0,0)$ would be added to our filter bank. If we choose to branch $(0)$, then both the blue path $(0,0)$ and the green path $(0,1)$ would be added to the filter bank. Additionally, 
	If we choose to branch the blue path $(0)$ at time $0$, then we add the green path $(0,1)$ as well as the blue path $(0,0)$. %The blue path $(0)$ branches into $(0,0)$ and $(0,1)$. 
	If we choose not to branch $(0)$, then only $(0,0)$ would be added to our filter bank. The red dotted-line depicts pruning; this would involve removing $(1,0)$ and all its children from the tree.}
	\label{fig:binary_tree}
\end{figure}

Limiting the growth of $\Fc$ can be done by deciding which branches to expand and which branches to prune. This sort of formulation lends itself very easily to an online formulation of the filter banks algorithm. In fact, it is more intuitive to think of the algorithm in an online fashion.

At time $t$, we choose to branch a filter only if it corresponds to
one of the most likely segmentations. %TODO: THIS WILL ALWAYS BE THE CASE? YES, BY ONE OF THE THEOREMS IN THE THEORY SECTION? THE ONE SAYING THE MAP UP TO A POINT CONDITIONED ON A CHANGE IS THE MAP PRIOR TO IT.
 By branching, we refer to exploring both the $0$ and $1$ branches. This %is depicted by the branching blue path in Figure~\ref{fig:binary_tree}. %This 
 is depicted by the blue and green lines in Figure \ref{fig:binary_tree}.
 Otherwise, we will merely extend the last segment of the segmentation, i.e. only follow the $0$ branch. Additionally, at time $t$, we prune any paths that have sufficiently low likelihood. That is, we remove the filter $f$ from $\Fc$ if $p_f < p_{\text{thres}}$, where $p_{\text{thres}}$ is an algorithm parameter. This is depicted by the red dotted line in Figure \ref{fig:binary_tree}.

Finally, we can exhibit our algorithm. It is presented in Figure \ref{alg:oedfb}.

%TODO: WE NEED A FRAME OR SOMETHING AROUND THE ALG.
\begin{figure}
\begin{algorithmic}[1]
	% \KwIn{input}
	% \KwOut{output}

\State{	Initialize $t \leftarrow 0$, $f_0\leftarrow\delta_{f_0} = (0)$, $f_1\leftarrow\delta_{f_1} = (1)$, and $\Fc\leftarrow\{f_0,f_1\}$.}
\State{ Pick algorithm parameter $p_{\text{thres}}$.}
\While{ TRUE } 
\LineComment{Find the filters that correspond to the most likely }
\LineComment{segmentations given $y^t$.}
\State{$\Fc'\leftarrow\emptyset$.}
\For{$f \in \Fc$}
\If{$p_f = \max_{f' \in \Fc}p_{f'}$}
\State{Add a copy of $f$ to $\Fc'$.}
\EndIf
\EndFor
\LineComment{When available, update all filters in $\Fc$ with the}
\LineComment{new measurement.}
\State{ Wait for new measurement $y[t+1]$.}
\For{$f \in \Fc$} 
\State{Append $0$ to $\delta_f$.
Recalculate $\ubarb_f$ and $p_f$.}
\EndFor
\LineComment{Branch the filters corresponding to the most likely }
\LineComment{segmentations given $y^t$.}
\For{$f' \in \Fc'$} 
\State{Append $1$ to $\delta_{f'}$.}
\State{Recalculate $\ubarb_{f'}$ and $p_{f'}$.}
\State{Add $f'$ to $\Fc$.}
\EndFor
\LineComment{Prune elements from the filter bank that have }
\LineComment{unlikely segmentations.}
\For{$f \in \Fc$}
\If{$p_f < p_{\text{thres}}$.}
\State{Remove $f$ from $\Fc$.}
\EndIf
\EndFor
\State{$t \leftarrow t + 1$.}
\EndWhile
\end{algorithmic}
	\caption{Algorithm for Online energy disaggregation via filter banks (OEDFB).}	\label{alg:oedfb}
\end{figure}
%\subsection{Example}
%\label{sec:alg_ex}

As presented currently, the algorithm give in Figure \ref{alg:oedfb}
is a high-level algorithm. We now discuss a specific implementation.
%for a given form of probability density. 
First, we will add some assumptions.
\begin{assumption}
\label{ass:steady_state}
	In the true $\ub$, each segment has length greater than or equal to $N = \max\ \{n_1,n_2,\dots,n_D\}$.
\end{assumption}
This assumption places a minimum length of a segment for our piecewise constant input $\ub$; it asserts that each device is in steady-state before a device changes state.

Let $y_{\text{ss},i}(k^n,l)$ denote the steady state value of $y_i$, the power consumption for the $i$th device, in the $l$th segment of $k^n$. That is:
\begin{equation}
y_{\text{ss},i}(k^n,l) = \sum_{j = 0}^{n_i} b_{i,j} \bar u_i(l).
\end{equation}
%Let:
%\begin{equation}
%\bar \betab = 
%\begin{bmatrix}
%\sum_{j = 0}^{n_i} b_{1,j} 	& \sum_{j = 0}^{n_i} b_{2,j} &\dots &\sum_{j = 0}^{n_i} b_{D,j}
%\end{bmatrix}^\top.
%\end{equation}
Then, let $y_{\text{ss}}(k^n,l)$ denote the steady-state value of $y$ in the $l$th segment of $k^n$; thus:
\begin{equation}
y_{\text{ss}}(k^n,l) = \sum_{i = 1}^D y_{\text{ss},i}(k^n,l).
\end{equation}
$y_{\text{ss}}$ can be directly estimated from our observations $y$, independently of our values for $\ubarb$. We will assume $y_{\text{ss}}(k^n,l)$ is known from this point onward. Additionally, let $y_{\text{ss}}(k^n,0) = 0$.

Suppose we are given a segmentation $\delta$. We will convert the estimation of $\ubarb$ into the estimation of the change in $\ubarb$. 
To such end, 
define $\Delta \ubarb(l) = \ubarb(l)-\ubarb(l-1)$, the change in input from segment $l-1$ to segment $l$, with $\Delta \ubarb(0) = 0$. 
%define $\Delta u_i[t] = u_i[t] - u_i[t-1]$, the change in the input to device $i$ from time $t-1$ to time $t$, with $\Delta u_i[0] = u_i[0]$, and similarly define:
%\begin{equation}
%	\Delta \ub[t] = 
%	\begin{bmatrix}
%		\Delta u_1[t] 	& \Delta u_2[t] & \dots & \Delta u_D[t]
%	\end{bmatrix}^\top.
%\end{equation}
%Note that $\Delta \ub[t] = 0$ if and only if $\delta[t] = 0$. 
%%%%%%%%%%%%%%%%%
%\begin{equation}
%	\zetab_i[t] =
%	\begin{bmatrix}
%		\Delta u_i[t] I(t > k_{l-1})	& u_i[t-1] & \dots & u_i[t-n_i]
%	\end{bmatrix}^\top.
%\end{equation}
%%%%%%%%%%%%%%%%%
%define $\Delta \ubarb(l) = \ubarb(l)-\ubarb(l-1)$, the change in input from segment $l-1$ to segment $l$, with $\Delta \ubarb(0) = 0$. Define:
Then, with Assumption \ref{ass:steady_state}, linearity implies that our dynamics take the following form:
\begin{align}
\begin{split}
	y[t] - y_{\text{ss}}(k^n,l-1) = \Lc_{k_{l-1}+1}(\Delta \ubarb(l), t) + e[t] \\
	\quad \text{ for } k_{l-1} < t \leq k_l,
\end{split}
\end{align}
where $\Lc_{k_l-1}(\Delta \ubarb(l), t)$ is the value of the zero-state step response at time $t$ of the aggregated system model in  \eqref{eq:agg_model1} to a step of $\Delta \ubarb(l)$ beginning at time $k_{l-1}+1$. Note that this linear function can easily be calculated from $\betab$.

The essential point of this equation is that, since all the devices are in steady state at the beginning of the $l$th segment, the actual values of $\ubarb(l-1)$ and $\ubarb(l)$ do not matter; the dynamics depend only on the change $\Delta \ubarb(l)$. Thus, we can estimate $\Delta \ubarb(l)$ separately for each segment $l$.
%\begin{equation}
%	\zetab[t] =
%	\begin{bmatrix}
%		\zetab_1[t]^\top 	& \zetab_2[t]^\top & \dots & \zetab_D[t]^\top
%	\end{bmatrix}^\top.
%\end{equation}

Furthermore, we consider the following prior. Suppose we have a bound on how much the input can change from segment to segment. That is, we know $\Delta \ub_{\text{min}}, \Delta \ub_{\text{max}}$ such that $\Delta \ub_{\text{min}} \leq \Delta \ubarb(l) \leq \Delta \ub_{\text{max}}$ for all $l$. Furthermore, $p(\ubarb \vert \delta)$ is a uniform distribution within these bounds.

Finally, if the noise term in  \eqref{eq:agg_model1} is Gaussian white noise with fixed variance $\sigma^2$, then the calculations of $\ubarb_f$ and $p_f$ are relatively straightforward. Let $y^l$ denote the portion of $y$ in segment $l$:
\begin{equation}
y^l = 
\begin{bmatrix}
y[k_{l-1}+1] & y[k_{l-1}+2] & \dots & y[k_l]
\end{bmatrix}^\top.
\end{equation}
By a slight abuse of notation, simply let $\Lc(\ub)$ denote the zero-state response of  \eqref{eq:agg_model1} to a step of $\ub$. We can find $\Delta \ubarb(l)$ by solving the following least-squares problem:
\begin{align}
\begin{split}
	\min_{\Delta \ubarb}\ 	& \| y^l - y_{\text{ss}}(k^n,l-1) - \Lc(\Delta \ubarb) \|_2^2 \\
	\subjto\ 	& \Delta \ub_{\text{min}} \leq \Delta \ubarb \leq \Delta \ub_{\text{max}}.
\end{split}
\end{align}
This will give us $\ubarb_f$. Let:
\begin{equation}
e[t] = y[t] - y_{\text{ss}}(k^n,l-1) - \Lc(\Delta \ubarb(l)) \text{ for } k_{l-1} < t \leq k_l
\end{equation}
We can also calculate:
\begin{equation}
p_f = c p(\delta) \prod_{s = 0}^t \exp\left(-\frac{e[s]^2}{2\sigma^2}\right)
\end{equation}
where $c$ is a constant that is independent of $\delta$ and $\ubarb$.

%% file: fig/binaryTree.tex
\begin{tikzpicture}[>=stealth',shorten >=1pt,auto,node distance=1.25 cm, scale = 1.25, transform shape]
  \draw (-1.75,0) node[fill=blue!25, circle, minimum size=0.5cm] (A) {\footnotesize root} (4,0);
  \draw (0.5,1) node[fill=blue!50, circle, minimum size=0.5cm] (B1) {$0$};
  \draw (0.5,-1) node[fill=blue!50, circle, minimum size=0.5cm] (B2) {$1$};
  \draw (2.75, 1.65) node[fill=blue!50, circle, minimum size=0.5cm] (C1) {$0$};
  \draw (2.75, 0.65) node[fill=blue!50, circle, minimum size=0.5cm] (C2) {$1$};
  \draw (2.75, -0.65) node[fill=blue!50, circle, minimum size=0.5cm] (C3) {$0$};
  \draw (2.75, -1.65) node[fill=blue!50, circle, minimum size=0.5cm] (C4) {$1$};
  \draw (-2.25, -2.5) node (D1) {};
  \draw (3.5, -2.5) node (D2) {};
  \path[->, thick] (D1) edge (D2);
  \draw (3.5, -2.5) node[below] (t) {$t$};
  \draw (0.5, -2.75) node (t1) {};
  \draw (0.5, -2.25) node (t2) {};
  \path (t1) edge (t2);
  \draw (2.75, -2.75) node (t3) {};
  \draw (2.75, -2.25) node (t4) {};
  \path (t3) edge (t4);
  \draw (2.75, -2.6) node[below] {$1$};
  \draw (0.5, -2.6) node[below] {$0$};
  \path[->, dblue, very thick] (A) edge (B1);
  \path[->, dblue, very thick] (B1) edge (C1);
  \path[->, dgreen, very thick] (B1) edge (C2);
  \path[->, very thick,gray] (A) edge (B2);
  \path[->, very thick,gray] (B2) edge (C3);
  \path[->, very thick, gray] (B2) edge (C4);
  \draw (1.15, -0.55) node (dd1) {};
  \draw (2.0, -1.1) node (dd2) {};
  \path[densely dashed, red, very thick] (dd1) edge (dd2);
  \draw (3.75, 2) node (D1) {};
  \draw (3.75, 1.3) node (D2) {};
  \path[->, gray!75, very thick, densely dotted] (C1) edge (D1);
  \path[->, gray!75, very thick, densely dotted] (C1) edge (D2);
  \draw (3.75, 1) node (D3) {};
  \draw (3.75, 0.3) node (D4) {};
  \path[->, gray!75, very thick, densely dotted] (C2) edge (D3);
  \path[->, gray!75, very thick, densely dotted] (C2) edge (D4);
  \draw (3.75, -0.3) node (D5) {};
  \draw (3.75, -1) node (D6) {};
  \path[->, gray!75, very thick, densely dotted] (C3) edge (D5);
  \path[->, gray!75, very thick, densely dotted] (C3) edge (D6);
  \draw (3.75, -1.3) node (D7) {};
  \draw (3.75, -2) node (D8) {};
  \path[->, gray!75, very thick, densely dotted] (C4) edge (D7);
  \path[->, gray!75, very thick, densely dotted] (C4) edge (D8);
  \draw[fill=black] (3.75,0) circle(0.25ex);
  \draw[fill=black] (3.95,0) circle(0.25ex);
  \draw[fill=black] (4.15,0) circle(0.25ex);

  %\fill[blue!25] (2,2) circle (1ex);
  %\draw node[fill=white, circle, right of=A,  minimum size=0.5cm, node distance=1.5 cm] (B1) {};
  %\draw node[fill=blue!25, circle, above of=B1,  minimum size=0.5cm] (B2) {$0$};
  %\draw node[fill=blue!25, circle, below of=B1,  minimum size=0.5cm] (B3) {$1$};
  %\draw node[fill=white, circle, right of=B2,  minimum size=0.5cm] (C1) {};
  %\draw node[fill=blue!25, circle, above of=C1,  minimum size=0.5cm, node distance=0.65 cm,] (C2) {$0$};
  %\draw node[fill=blue!25, circle, below of=C1, node distance=0.65 cm,  minimum size=0.5cm] (C3) {$1$};
  %\draw node[fill=white, circle, right of=B3,  minimum size=0.5cm] (C4) {};
  %\draw node[fill=blue!25, circle, above of=C4,  minimum size=0.5cm, node distance=0.65 cm,] (C5) {$0$};
  %\draw node[fill=blue!25, circle, below of=C4, node distance=0.65 cm,  minimum size=0.5cm] (C6) {$1$};
 %\path[->] (A) edge  (B2);
 %\path[->] (B2) edge (C2);
  %child[grow=right]{ node{1}}
  %child[grow=right]{node{0}};
     % child { node {1}}
     % child { node{0}}};
\end{tikzpicture}

%% file: theoretical.tex
%%%%%%%%%%%%%%%%%%%%%%%%%%%%%%%%%%%%%%%%%%%%%%%%%%%%%%%%%%%%%%%%%%%%%%%%%%%%%%%%%%%%%%%%%%%%%%%%%%%%%%%%%%%%%%%%%%%%%%%%%%%%%%%
One of the benefits of our framework is that it allows us to leverage results from adaptive filtering. In this section, we prove theorems relating to the algorithm presented in Section~\ref{sec:implementation}.

Let $\widehat {\delta^t}$ denote any segmentation such that:
\begin{equation}
\label{eq:map_def}
	p(\widehat {\delta^t} \vert y^t) = \max_{\delta \in \{0,1\}^t} p(\delta \vert y^t),
\end{equation}
for any $t \in \{0,1,\dots,T\}$. That is, $\widehat {\delta^t}$ denotes a maximum a posteriori  estimate of the segmentation $\delta$ at time $t$. We can now apply the following result:
\begin{theorem} (\emph{Optimality of partial MAP estimates} \cite{Gustafsson2000})
	\label{th:map}
	Let $t$ be any arbitrary time in $\{0,1,\dots,T\}$, and let $t_0$ be any time such that $0 \leq t_0 \leq t$. Let $\delta$ be any binary sequence of length $t$ such that $\delta[t_0] = 1$. Let $\delta_1$ denote the first $t_0-1$ elements of $\delta$ and $\delta_2$ denote the last $t - t_0$ elements of $\delta$. That is: $\delta = (\delta_1,1,\delta_2)$.
	
	If Assumption \ref{ass:steady_state} holds and if $\Delta \ubarb(l)$ and $\Delta \ubarb(m)$ are independent given $\delta$ for $l \neq m$, then:
	\begin{equation}
		p(\delta \vert y^t) \leq p((\widehat {\delta^{t_0-1}}, 1,\delta_2) \vert y^t)
	\end{equation}
\end{theorem}
\begin{proof}
	Note that our hypotheses together imply that $\Delta \ubarb(l)$ and $\Delta \ubarb(m)$ are independent given $\delta$ and $y^t$ for $l \neq m$. Thus:
	\begin{equation}
		\begin{array}{rcl}
			p(\delta \vert y^t) 	& = & p((\delta_1,1,\delta_2) \vert y^t, \delta[t_0] = 1) p(\delta[t_0] = 1 \vert y^t) \\
							& = & p(\delta_1 \vert y^t, \delta[t_0] = 1) p(\delta_2 \vert y^t, (\delta_1,1)) \cdot \\
							&& \quad p(\delta[t_0] = 1 \vert y^t) \\
							& = & p(\delta_1 \vert y^{t_0-1}) p(\delta_2 \vert y^t, \delta[t_0] = 1) p(\delta[t_0] = 1 \vert y^t) \\
							& \leq & p(\widehat {\delta^{t_0-1}} \vert y^{t_0-1}) p(\delta_2 \vert y^t, \delta[t_0] = 1) \cdot \\
							&& \quad p(\delta[t_0] = 1 \vert y^t) \\
							& = & p(\widehat {\delta^{t_0-1}} \vert y^t) p(\delta_2 \vert y^t, \delta[t_0] = 1) \cdot \\
							&& \quad p(\delta[t_0] = 1 \vert y^t) \\
							& = & p((\widehat {\delta^{t_0-1}}, 1,\delta_2) \vert y^t)
		\end{array}
	\end{equation}
	where, by a slight abuse of notation, $p(\delta_1 \vert \delta_2)$ denotes the likelihood that the first $t_0-1$ elements of the true $\delta$ are equal to $\delta_1$ given that the last $t - t_0$ elements are equal to $\delta_2$.
	
	The first and second equalities utilize Bayes' law. Causality
        and independence of segments imply that $\delta_2$ does not
        depend on $\delta_1$ given that $\delta[t_0] = 1$ and that
        $\delta_1$ does not depend on later measurements given that
        $\delta[t_0] = 1$. This gives us the third equality. The
        inequality follows from the definition of $\widehat
        {\delta^{t_0 - 1}}$, given in  \eqref{eq:map_def}. The final equalities are similar to the first equalities.
\end{proof}

This theorem implies that, conditioned on a change at time $t_0$, the
MAP sequence at time $t$ must begin with the MAP sequence at time
$t_0$. Also, note that under Assumption \ref{ass:steady_state}, we have that $\Delta \ubarb(l)$ and $\Delta \ubarb(m)$ are independent given $\delta$ for $l \neq m$.

We can now assert the following claims about our algorithm:
\begin{theorem}
\label{th:branch}
(\emph{Optimality of the proposed algorithm's branching policy} \cite{Gustafsson2000})
	Consider the algorithm given in Figure \ref{alg:oedfb} with $p_{\text{thres}} = 0$. Suppose Assumption \ref{ass:steady_state} holds and $\Delta \ubarb(l)$ and $\Delta \ubarb(m)$ are independent given $\delta$ for $l \neq m$.
	
	Fix any time $t$, and let $\Fc$ be the filter bank at time $t$. Then, there exists an $f \in \Fc$ such that $\delta_f = \widehat {\delta^t}$.
\end{theorem}
\begin{proof}
If $\widehat {\delta^t}[s] = 1$ for some $s$, then the first $s-1$ elements of $\widehat {\delta^{t}}$ is a MAP estimate at time $s - 1$. This follows from Theorem \ref{th:map}. This means that, at time $s$, we only need to branch the most likely segmentations.
\end{proof}

Theorem \ref{th:branch} states that, in the case of no pruning, any MAP estimate will still be present in the filter bank. In other words, maximizing over the reduced set of filters in the filter bank will be equivalent to maximizing over every single possible segmentation.

This theorem also gives rise to our corollary. First, let $(\widehat {\delta^t})^s$ denote the first $s$ elements of $\widehat {\delta^t}$. Then:
\begin{corollary}
\label{th:prune}
(\emph{Optimality of proposed algorithm's pruning policy})
	Consider the algorithm given in Figure \ref{alg:oedfb} with $p_{\text{thres}} > 0$. Suppose Assumption \ref{ass:steady_state} holds and $\Delta \ubarb(l)$ and $\Delta \ubarb(m)$ are independent given $\delta$ for $l \neq m$.
	
	Fix any time $t$, and let $\Fc$ be the filter bank at time $t$. If $p( (\widehat {\delta^t})^s \vert y^s ) \geq p_{\text{thres}}$ for all $0 \leq s < t$, then there exists an $f \in \Fc$ such that $\delta_f = \widehat {\delta^t}$.
\end{corollary}
\begin{proof}
	If $p(\widehat{\delta^s} \vert y^s) \geq p_{\text{thres}}$ for all $0 \leq s < t$, then $\widehat {\delta^t}$ will never be pruned.
\end{proof}

Corollary \ref{th:prune} states a condition for when an MAP estimate will still be present in the filter bank at time $t$.

%% file: experiment.tex
%%%%%%%%%%%%%%%%%%%%%%%%%%%%%%%%%%%%%%%%%%%%%%%%%%%%%%%%%%%%%%%%%%%%%%%%%%%%%%%%%%%%%%%%%%%%%%%%%%%%%%%%%%%%%%%%%%%%%%%%
\subsection{Experimental setup}

To test our disaggregation method, we deployed a small-scale experiment. To collect data, we use the emonTx wireless open-source energy monitoring node from OpenEnergyMonitor\footnote{ {\tt{http://openenergymonitor.org/emon/emontx}}}. We measure the current and voltage of devices with current transformer sensors and an alternating current (AC) to AC power adapter. For each device $i$, we record the root-mean-squared (RMS) current $I_{\text{RMS}}^i$, RMS voltage $V_{\text{RMS}}^i$, apparent power $P_{\text{VA}}^i$, real power $P_{\text{W}}^i$, power factor $\phi_{\text{pf}}^i$, and a coordinated universal time (UTC) stamp. The data was collected at a frequency of 0.13Hz.

Our experiment focused on small devices commonly found in a residential or commercial office building. First, we recorded plug-level data $z_i$ for a kettle, a toaster, a projector, a monitor, and a microwave. 
%These devices were labeled $1,2,\dots,5$ respectively. 
These devices consume anywhere from 70W to 1800W. For each device, we fit a fifth-order FIR model as outlined in Section \ref{sec:dev_mod}.

Then, we ran an experiment using a microwave, a toaster, and a kettle operating at different time intervals. %These are devices 5, 2, and 1 respectively. 
These measurements form our ground truth $y_i$, and we also sum the signals to get our aggregated power signal $y = \sum y_i$. The individual plug measurements are shown in Figure \ref{fig:meas_exp}. It is worth commenting that the power consumption signals for individual devices are not entirely independent; one device turning on can influence the power consumption of another device. This coupling is likely due to the non-zero impedance of the power supply system. However, we found this effect to be negligible in our disaggregation algorithms.

\begin{figure}[ht]
	\begin{center}
	\includegraphics[width=\columnwidth]{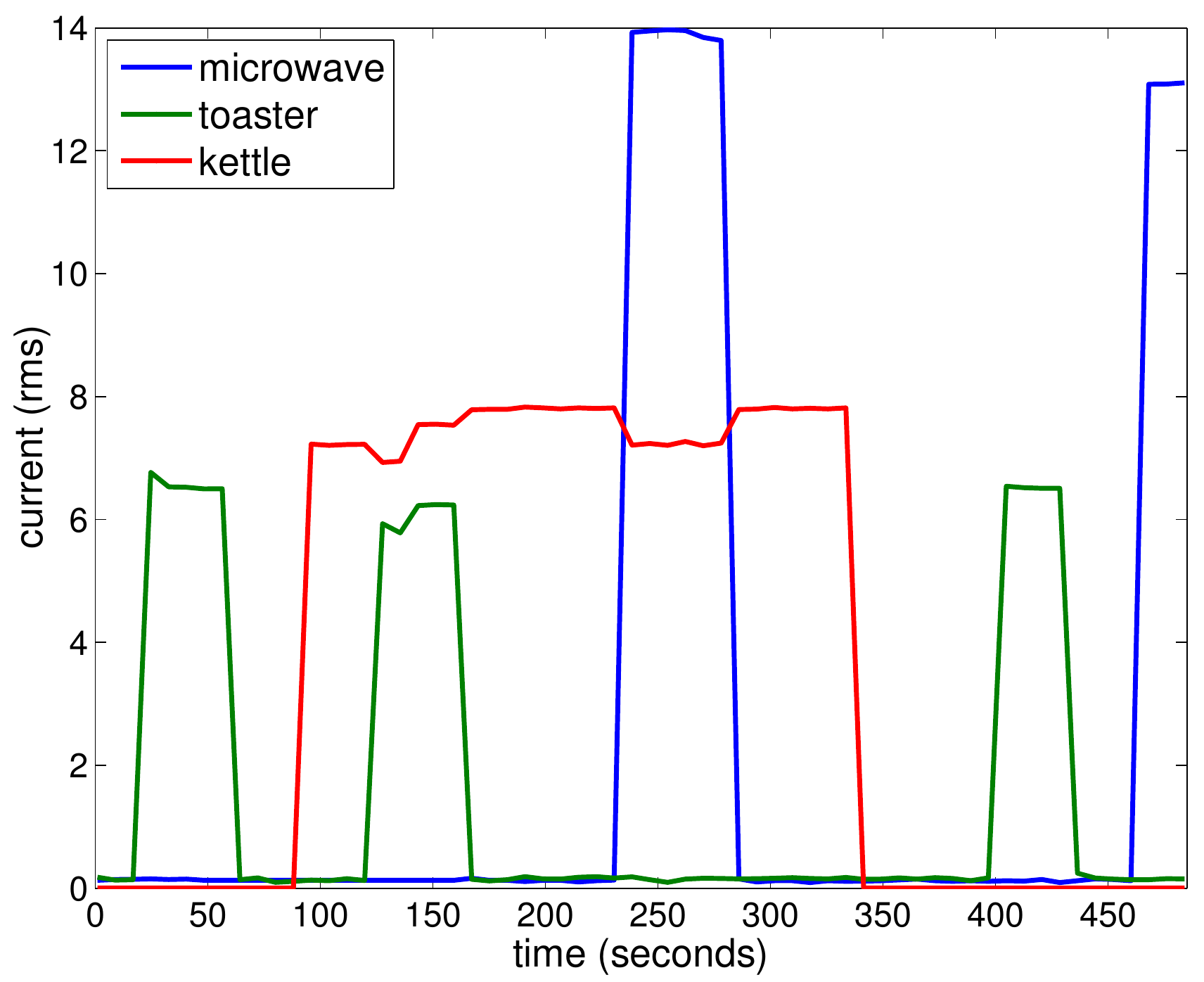}
	\end{center}
	\caption{The measurements of individual plug RMS currents.}
	\label{fig:meas_exp}
\end{figure}

%%%%%%%%%%%%%%%%%%%%%%%%%%%%%%%%%%%%%%%%%%%%%%%%%%%%%%%%%%%%%%%%%%%%%%%%%%%%%%%%%%%%%%%%%%%%%%%%%%%%%%%%%%%%%%%%%%%%%%%%
\subsection{Implementation details}

In practice, we observed that many devices seem to have different dynamics between when they switch on and when they switch off. For example, consider the root-mean-squared (RMS) current of a toaster in Figure \ref{fig:toaster}. There is an overshoot when the toaster switches on, but the the dynamics when the device shuts off do not exhibit the same behavior. In fact, in all of the devices we measured, we found that when a devices switches off, the power consumption drops down to a negligible amount almost immediately. That is, we do not observe any transients when a device turns off. We modify the models from Section \ref{sec:dev_mod} to encapsulate this observation.

%Small details need to be accounted for when actually implementing the algorithm.

\begin{figure}[ht]
	\begin{center}
	\includegraphics[width=\columnwidth]{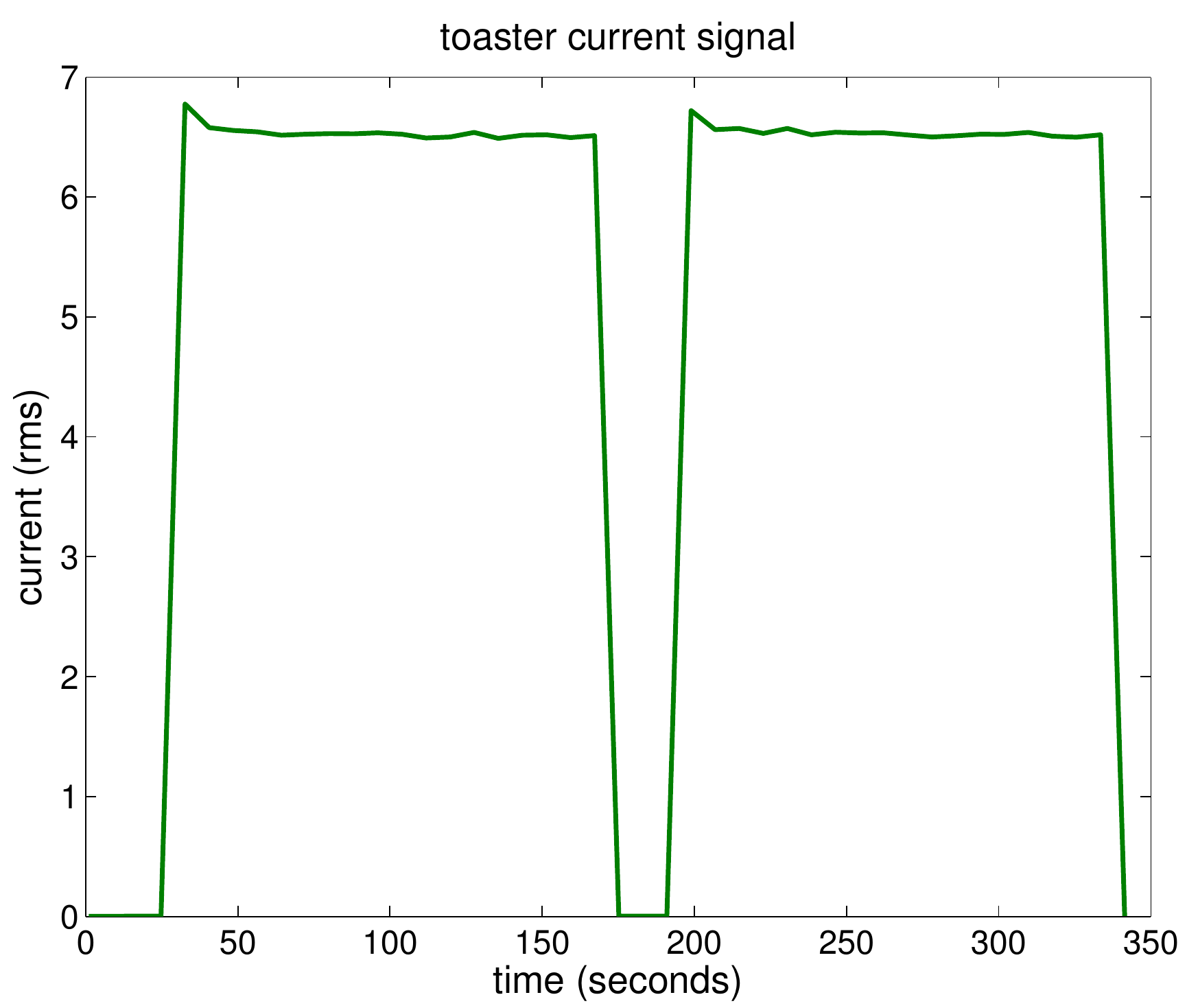}
	\end{center}
	\caption{The measured RMS current signal for a toaster. Note that the on-switches display overshoot while the off-switches do not.}
	\label{fig:toaster}
\end{figure}

Several heuristics are used for pruning the binary tree depicted in Figure \ref{fig:binary_tree} that are specific to the task of disaggregation. First, we do not bother considering branches if the most likely segmentation explains the data sufficiently well. This greatly reduces the growth of the filter bank across time. Furthermore, we assume that at most one device switches on or off in any given time step. This unfortunately violates the assumptions of Theorem \ref{th:map}, but we find that it gives good results in practice. 
%Finally, we assume that if the power consumption increases, then a device is turni%ng on, and if the power consumption decreases, then a device is turning off.

%%%%%%%%%%%%%%%%%%%%%%%%%%%%%%%%%%%%%%%%%%%%%%%%%%%%%%%%%%%%%%%%%%%%%%%%%%%%%%%%%%%%%%%%%%%%%%%%%%%%%%%%%%%%%%%%%%%%%%%%
\subsection{Results}

The disaggregation results are presented in Figure \ref{fig:exp_inputs}. We can see that the segmentation is correctly identified. Visually, the results also line up well. 

\begin{figure}[ht]
	\begin{center}
	\includegraphics[width=\columnwidth]{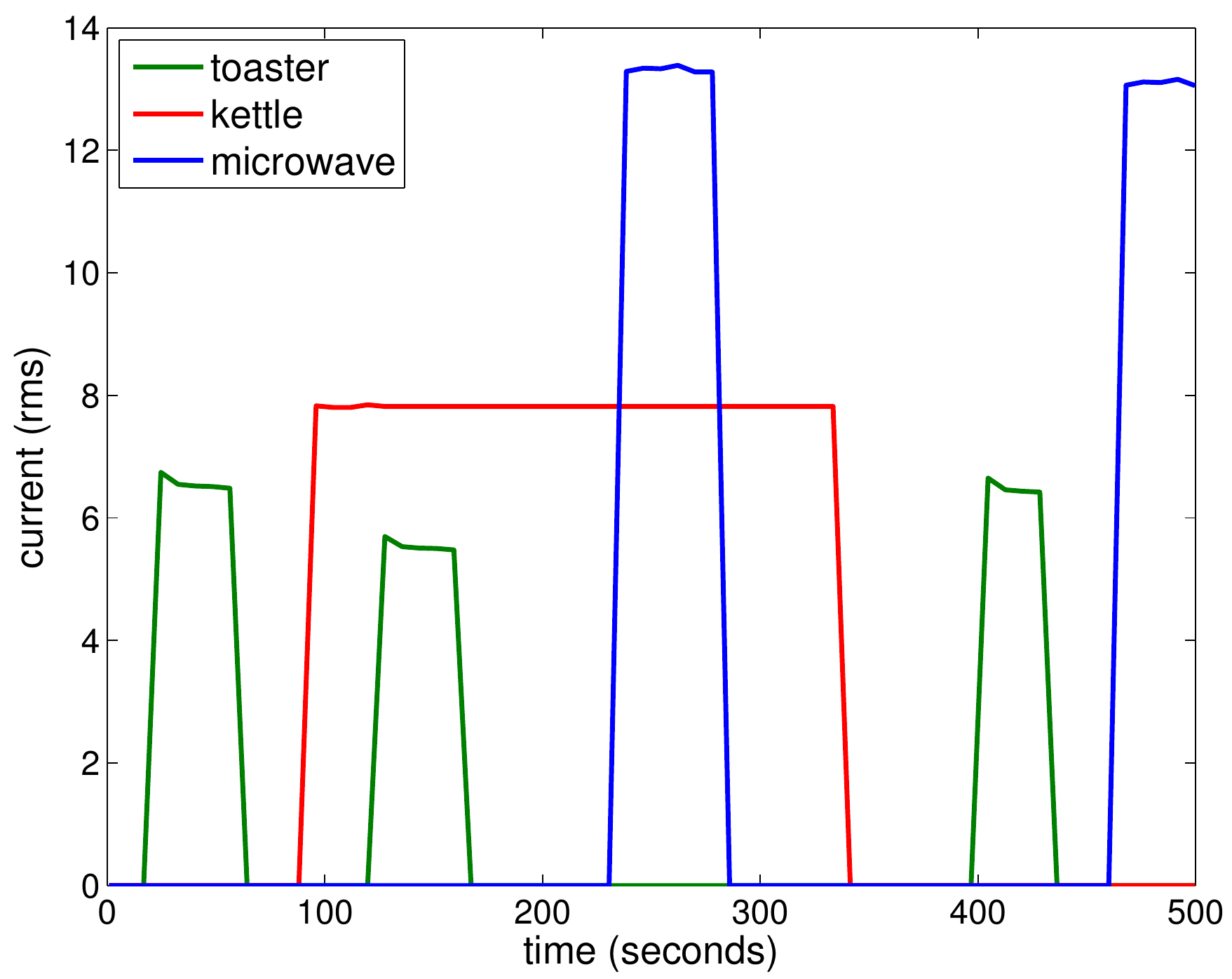}
	\end{center}
	\caption{The estimated power consumption signals of each device.}
	\label{fig:exp_inputs}
\end{figure}

We also note that it is not fair to compare results from our small-scale experiment with many of the methods mentioned in Section \ref{sec:background}. Most of the methods listed are unsupervised methods which do not have a training set of data \cite{Kolter2011,Shao2012,Kolter2012,Johnson2012}. Since these unsupervised methods do not learn from training data, they have many priors which must be tuned towards the devices in the library. Also, the sparse coding method in \cite{Kolter2010} requires a large amount of disaggregated data to build a dictionary.

%% file: conc.tex
In the work presented, we formalized the disaggregation problem within the filter banks framework. We provide an algorithm with guarantees on the recovery of the true solution given some assumptions on the data.

%disaggregation problem in an adaptive filtering framework.
%
%previous work by formalize the use of filter banks and dynamical models for disaggregation of whole building energy data. We devlop an algorithm and give guarantees on its convergence provided the data satisfies some assumptions. 
%In~\cite{Dong2013}, we introduce a dynamical systems approach to disaggregation of aggregate building energy data. In this paper, 

From the point of view of the utility company, the question of how to use this data to inform the consumer about their usage patterns and how to develop incentives for behavior modification %, such as demand shifting 
is still largely an open one, which we are currently studying.

Another largely open question is the one concerning privacy. Given that energy data can be disaggregated with some degree of precision, how does this affect the consumer's privacy?  The next natural step is to study how this data can be used in a privacy preserving way to improve energy efficiency. These privacy preserving policies may come in the form of selectively transmitting the most relevant data for a control objective, or incentive mechanisms for users to change their consumption behavior without direct transmission of their private information to the utility company. 
We are currently examining both approaches to the privacy issue.